\documentclass[runningheads]{llncs}
\usepackage[T1]{fontenc}
\DeclareFontShape{T1}{cmr}{m}{scit}{<-> ssub * cmr/m/sc}{}
\usepackage{silence}
\usepackage{mathtools}
\usepackage{array}
\usepackage{multirow} 
\usepackage{marvosym}
\usepackage{todonotes}
\usetikzlibrary{arrows.meta}

\usepackage{graphicx}
\usepackage[hidelinks]{hyperref}
\usepackage{color}

\usepackage{framed}
\def\shuffleProduct{\sqcup\mathchoice{\mkern-7mu}{\mkern-7mu}{\mkern-3.2mu}{\mkern-3.8mu}\sqcup{\mkern3.8mu}}
\newcommand{\IDLetter}[1]{\langle #1 \rangle} 
\newcommand{\bBlockLength}[1]{{\ell}_{#1}} 
\newcommand{\SchedulingBaseProblem}[0]{\textsc{Scheduling Of Jobs With Precedence Constraints}}

\begin{document}
\title{The Parameterized Complexity of Scheduling with Precedence Delays: \texorpdfstring{\\}{} Shuffle Product and Directed Bandwidth\thanks{An extended abstract of this paper appears in the proceedings of the 37th International Workshop on Combinatorial Algorithms, IWOCA 2026.}}
\titlerunning{Scheduling with Precedence Delays: Shuffle Product and Directed Bandwidth}
%
\author{Hans L. Bodlaender\inst{1}\orcidID{0000-0002-9297-3330} \and \\ Maher Mallem \inst{2}\orcidID{0000-0001-5654-1090} (\Letter)}
\authorrunning{H.~L.~Bodlaender \and M.~Mallem}
%
\institute{Department of Information and Computing Sciences, Utrecht University, the Netherlands \\
\email{h.l.bodlaender@uu.nl}\\ \and
 Inria, CNRS, ENS de Lyon, Université Claude Bernard Lyon 1, LIP, UMR 5668, 69342, Lyon cedex 07, France\\
\email{maher.mallem@ens-lyon.fr}}
\maketitle              
\begin{abstract}
In this paper, we study the parameterized complexity of several variants of scheduling with precedence constraints between jobs. Namely, we consider the single machine setting with delay values on top of the precedence constraints. Such scheduling problems are related to several decades-old problems with open parameterized complexity status, notably \textsc{Shuffle Product} and \textsc{Directed Bandwidth}. We obtain XNLP-completeness results for both problems, and derive implications to scheduling with minimum (resp. maximum) delays parameterized by the width of the directed acyclic graph giving the precedence constraints, and/or by the maximum delay value in the input. Regarding \textsc{Directed Bandwidth}, we also settle the case of trees by showing XNLP-completeness parameterized by the target value. Beyond these results, we believe that \textsc{Shuffle Product} is an unusual and promising addition to the list of XNLP-complete problems.

\end{abstract}

\section{Introduction}
In this paper, we study the parameterized complexity of several scheduling problems where there are precedence constraints with delays between jobs. 
Such delays are handy in that they can model a wide range of real-life problems. For instance, minimum delays mimic how instructions are handled in pipelined processors~\cite{BrunoJonesSo:80:Deterministic}, and maximum delays can model products of short shelf life. Even theoretically speaking, several decades-old problems with open parameterized complexity status like \textsc{Unary Bin Packing}~\cite{JansenKratschMarx:13:Bin} and \textsc{Directed Bandwidth}~\cite{Bodlaender:21:Parameterized} are equivalent to scheduling settings with precedence delays - more on this in the later sections.

Several of the problems studied in this paper have the same complexity: they are known to be in the class XP --- that is, for each fixed value of the parameter, there is a polynomial time algorithm, but the exponent of the polynomial grows when the parameter grows, and the problems are not known to be fixed-parameter tractable. For many problems, the precise
complexity is unknown. In this paper, we show that several parameterized scheduling problems with precedence constraints and delays are complete for a complexity class now known as~XNLP.

The class XNLP was originally introduced by Elberfeld et al.~\cite{ElberfeldST15}, and renamed to XNLP by Bodlaender et al.~\cite{BodlaenderGNS22a}. In~\cite{BodlaenderGNS22a}, several problems were shown to be complete for this class, including 
\SchedulingBaseProblem{} with both the number of machines and the width of the precedence graph as parameters --- the first example where a scheduling problem was shown to be XNLP-complete.
Other recent work that shows that scheduling problems are complete
for the class XNLP or for the related class XALP~\cite{BodlaenderGJPP22a} can be found in \cite{BodlaenderHL25,Mallem:24:Parameterized}.
Showing that a problem is XNLP-complete is interesting for several reasons: it pinpoints the exact complexity class of the problem
and, assuming a conjecture by Pilipczuk and Wrochna~\cite{PilipczukW18}, we have that there is no algorithm in XP for the problem that uses
only $f(k)n^{O(1)}$ space --- indeed, about all known XP algorithms for these problems employ dynamic programming and have space usage of $\Omega\left(n^{f(k)}\right)$ for $f$ a divergent function. Also, XNLP-hardness implies hardness for all classes $W[i]$ with $i$ positive integer (e.g., see Lemma~$2$ in~\cite{BodlaenderGNS22a}).

Among the several parameterized problems whose complexity we resolve by showing XNLP-completeness, we have the \textsc{Shuffle Product} problem. Given a target word~$t$ and source words $s_1, \ldots, s_k$, one asks whether $t$ can be obtained by interleaving the letters of $s_1, \ldots, s_k$, while keeping the order of the letters from the original words.
We strengthen the W[2]-hardness proof for \textsc{Shuffle Product} from~\cite{vanBevern16} to XNLP-completeness (Section~\ref{subsection:sp_XNLP_completeness}). We believe that this
result is of independent interest, as \textsc{Shuffle Product} can be used as the starting point for several XNLP-hardness proofs. From this
result, in Section~\ref{subsection:sp_consequences}, we notably derive XNLP-completeness of single machine scheduling with precedence constraints, constant minimum delays and processing time at most two parameterized by the width of the directed acyclic graph.

In Section~\ref{section:directed_bandwidth}, we show XNLP-completeness of a scheduling problem with precedence constraints and maximum delays, namely the single machine scheduling of unit-time jobs under precedence constraints with equal maximum delays. This problem is equivalent to the directed variant of the well studied \textsc{Bandwidth} problem. Formulating it as a graph problem, we are given in the \textsc{Directed Bandwidth} problem a directed acyclic graph, and ask for a topological ordering $f$ with for each arc $vw$, $f(v) < f(w)$ and $f(w)-f(v)\leq k$ for some target value $k$. We give two XNLP-completeness results for \textsc{Directed Bandwidth} in Section~\ref{section:directed_bandwidth}, namely for general directed acyclic graphs with both the target bandwidth and the width as parameters, and for directed trees with the target bandwidth as parameter.

\begin{table}[tbp]
    \centering
    \setlength\arrayrulewidth{1pt}
    \setlength\extrarowheight{3pt} 
    \begin{tabular}{|c|c|c|c|}
        \hline
        \emph{Problem} & \emph{Section} & \emph{Parameter} & \emph{Result} \\ \hline
        \multirow{2.25}{*}{\hyperref[definition:shuffle_product_problem]{\textsc{(Binary) Shuffle Product}}} & \multirow{2.25}{*}{\ref{subsection:sp_XNLP_completeness}} & number of & \multirow{2.25}{*}{XNLP-complete} \\ 
        & & source words & \\ \hline
        \multirow{4.25}{*}{\hyperref[definition:directed_bandwidth]{\textsc{Directed Bandwidth}}} & \multirow{2.25}{*}{\ref{subsection:dbdags}} & target value & XNLP-complete for \\ 
        & & plus width & directed acyclic graphs \\ \cline{2-4}
        & \multirow{2.25}{*}{\ref{subsection:dbtrees}} & \multirow{2.25}{*}{target value} & XNLP-complete for \\ 
        & & & directed trees \\ \hline
    \end{tabular}
    
    \begin{tabular}{c}
        \\
    \end{tabular}
    \caption{Summary of our main results. Connections between \textsc{Shuffle Product} and scheduling with minimum precedence delays are discussed in Section~\ref{subsection:sp_consequences}. \textsc{Directed Bandwidth} is equivalent to single machine scheduling of unit-time jobs with equal-length maximum precedence delays - see the beginning of Section~\ref{section:directed_bandwidth}.}
    \label{table:results}
\end{table}

Our main results are summarized in Table~\ref{table:results}. In many cases, we improve upon or complement existing results. We delay the discussion of related literature to the separate sections. Figure~\ref{fig:problem_reductions} summarizes the XNLP-hardness reductions proposed in this paper.

\section{Preliminaries}\label{section:preliminaries}

\subsection{Definitions}\label{subsection:definitions}

    Given two integers $a,b$, $[a,b]$ denotes set $\{a, a+1, \ldots, b\}$.
    
    \paragraph{Graphs.} 
    
    We assume the reader to be familiar with standard notions from graph theory. Throughout the paper, given a graph $G$, $n$ denotes the number of vertices of $G$.
       
    We consider two types of directed trees: downwards directed trees, where each arc is directed from a vertex to one of its children, and upward directed trees, where each arc is directed from a vertex to its parent. Note that a hardness proof for one of these two trees directly transforms to a hardness proof for the other type, by reversing the direction of all arcs.
    
    \paragraph{Scheduling.} 
    
    The goal of scheduling is to complete a set of jobs over time under a set of constraints. Unless otherwise stated, $n$ is the number of jobs, $\mathcal{J} = [1,n]$ is the set of jobs and $G_{prec} = (\mathcal{J}, A)$ is a directed acyclic graph representing precedence constraints between jobs. Given two jobs $i,j \in \mathcal{J}$, if $ij$ is an arc in $G_{prec}$ then job~$i$ must be completed before job~$j$ is started. For the sake of simplicity, we restrict ourselves to the parallel identical machine setting, where $m$~machines are available to execute jobs. Each machine can execute at most one job at a time. In our base setting, every job takes unit time to be processed. The formal definition is the following:
    
    \begin{verse}
        \SchedulingBaseProblem \\
        \textbf{Given:} positive integers $n,m,D$ and a directed acyclic graph $G_{prec}$.  \\
        \textbf{Question:} is there a schedule completing all jobs by time $D$ and meeting all precedence and machine constraints? In other words, is there a function $f: \mathcal{J} \rightarrow [0,D-1]$ assigning start times to jobs such that:
        \begin{itemize}
            \item[($i$)] for all $t \in [0, D-1]$, $|\{i \in \mathcal{J}, f(i) = t\}| \le m$, and
            \item[($ii$)] for every arc $ij$ in $G_{prec}$, $f(i) + 1 \le f(j)$?
        \end{itemize}
    \end{verse}
    
    Natural extensions include giving two properties to each job~$j$: a positive integer $p_j$ representing its processing time, and a positive integer $size_j$ representing the number of required machines to process job~$j$. In the base setting, both $p_j$ and $size_j$ are equal to one. If values $size_j$ and/or $p_j$ can be non-unit, then condition~($i$) (resp.~($ii$)) becomes: $\left(\sum_{j \in \mathcal{J}, f(j) \le t \le f(j)+p_j-1} size_j\right) \le m$ (resp. $f(i) + p_i \le f(j)$).
    
    In order to denote these problem variants in a concise way, we use the standard three-field notation $\alpha | \beta | \gamma$ introduced in~\cite{GrahamLawlerLenstra:79:Optimization}. Fields~$\alpha$, $\beta$ and $\gamma$ respectively describe the machine environment, the job properties and the objective function. In this paper, the latter will always be denoted by `$C_{max} \le D$', meaning that all jobs must be completed by time~$D$. In field~$\alpha$, the parallel identical machine setting is denoted by `$P$'. If $m$ is a constant then we denote `$Pm$'; if $m=1$ then we are in the single machine setting, which is denoted by `$1$'. In field~$\beta$, `$prec$' denotes the presence of directed acyclic graph $G_{prec}$ in the input. If there is no mention of $size_j$ (resp. $p_j$), then by default they are equal to $1$ (resp. they can be any positive integer). For instance, \SchedulingBaseProblem{} is denoted by $P|prec, p_j=1|C_{max} \le D$.
    
    Precedence delays specify additional waiting time constraints on top of existing precedence relations. Given an arc $ij$ in $G_{prec}$, we say that there is a minimum (resp. maximum, exact) precedence delay of value $\ell_{i,j}$ when job~$j$ can only start at least (resp. at most, exactly) $\ell_{i,j}$ time units after job~$i$ is completed. In particular, note that minimum precedence delays of value zero have a matching requirement to regular precedence constraints. When scheduling problems allow for precedence delays, we extend the three-field notation the same way as in~\cite{MallemHanenMunier:25:Single}. Namely, in field~$\beta$, the presence of minimum (resp. maximum, exact) precedence delays is denoted by `$prec(\ell^{min})$' (resp. `$prec(\ell^{max})$', `$prec(\ell^{ex})$'). In this case, all arcs in~$G_{prec}$ are equipped with a delay of the respective nature.
    
    In this paper, we focus on two parameters for scheduling problems: the \emph{width} of the instance, which is defined as the width of~$G_{prec}$ (see the beginning of Section~\ref{section:directed_bandwidth}), and, in the presence of precedence delays, the \emph{maximum delay value} in the input.
    
    \paragraph{XNLP.} 
    
    We assume the reader to be familiar with the most standard notions of parameterized complexity; otherwise, e.g., see~\cite{Cyganbook}. In this paper, we focus on showing problems to be complete for the class XNLP.
    
    XNLP is the class of parameterized problems, that can be solved with a non-deterministic
    Turing Machine in $f(k)n^{O(1)}$ time and $O(f(k)\log n)$ space, for a computable function $f$, with $k$ the parameter and $n$ the input size. A \emph{parameterized logspace} reduction from parameterized problem $A$ to parameterized problem $B$ is an algorithm, that, when
    given an input $(x,k)$ gives an output $(x',k')$, with the following properties:
    (1) $A(x,k)\Leftrightarrow B(x',k')$; (2) $k'\leq g(k)$ for a computable function $g$;
    (3) the algorithm uses $f(k) + O(\log n)$ space for a computable function $f$; (4) the algorithm uses $n^{O(1)}$ time.

    \begin{figure}[tbp]
        \centering
        \usetikzlibrary{decorations.pathmorphing,calc}
        \begin{tikzpicture}
            \newcommand{\xLeft}{-3}
            \newcommand{\xRight}{3}
            \newcommand{\yStep}{-1.2}
            \tikzstyle{problem}=[draw,align=center,rectangle,minimum height=.55cm]
            \node[problem] (NNCCM) at (0,0) {\textsc{NNCCM Acceptance}};
            \node[problem] (SP) at (\xLeft,\yStep) {\textsc{(Binary) Shuffle Product}};
            \node[problem] (Parallel) at (\xLeft,2*\yStep) {Scheduling with parallel \\ identical machines};
            \node[problem] (Min) at (\xLeft,3.1*\yStep) {Scheduling with minimum \\ precedence delays};
            \node[problem] (DB) at (\xRight,\yStep) {\textsc{Directed Bandwidth}};
            \node[problem] (Max) at (\xRight,2*\yStep) {Scheduling with maximum \\ precedence delays};
            \tikzstyle{reduction}=[-{Latex[width=2mm]}]
            \draw[reduction] let \p1=($(SP)-(NNCCM)$) in
                (NNCCM) -- +(\x1,0) -- (SP) node[right, pos=0.5] {\ref{subsection:sp_XNLP_completeness}};
            \draw[reduction] (SP) -- (Parallel) node[right,midway] {\ref{subsection:sp_consequences}};
            \draw[reduction] (Parallel) -- (Min) node[right,midway] {\ref{subsection:sp_consequences}};
            \draw[reduction] let \p2=($(DB)-(NNCCM)$) in
                (NNCCM) -- +(\x2,0) -- (DB) node[right, pos=0.5] {\ref{subsection:dbdags}, \ref{subsection:dbtrees}};
            \draw[reduction] (DB) -- (Max);
            \draw[reduction] (Max) -- (DB);
        \end{tikzpicture}
        \caption{Reductions between XNLP-complete problems in this paper. See the corresponding sections for the exact definitions of the problem variants.}
        \label{fig:problem_reductions}
    \end{figure}
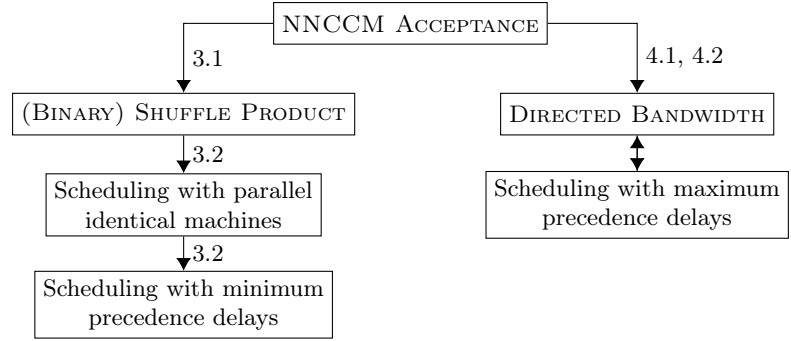
    
    The class XNLP was first introduced by Elberfeld et al.~\cite{ElberfeldST15}, and was renamed
    to XNLP in~\cite{BodlaenderGNS22a}. Also, Bodlaender et al.~\cite{BodlaenderGNS22a} introduced
    a problem that they named \textsc{NNCCM Acceptance}, and showed it to be XNLP-complete. This problem, defined below, is used as starting point for our reductions.
    
    In an instance of \textsc{NNCCM Acceptance} we are given: an integer $k$,
    that tells the number of counters we have; an integer $n$ (the maximum value of each counter); and a sequence of $r$ \emph{checks}. Each check is a 4-tuple $(c_1,n_1,c_2,n_2) 
    \in [1,k] \times [0,n] \times [1,k] \times [0,n]$. These describe a (somewhat artificial but useful for hardness proof) machine model, that works as follows.
    \begin{itemize}
        \item There are $k$ counters, each with an initial value of $0$.
        \item Each of $r$ rounds consists in the following two steps. In round $i\in [1,r]$:
        \begin{enumerate}
            \item In this non-deterministic step, some, one or none of the counters are increased, but never to a value larger than $n$.
            \item The $i$th check is done: if the $c_1$st counter has value $n_1$ and the
            $c_2$th counter has value $n_2$, then the machine halts and rejects.
        \end{enumerate}
        \item If the machine did not reject in any of the $r$ rounds, it accepts.
    \end{itemize}
    
    \begin{verse}
        \textsc{NNCCM Acceptance} \\
        \textbf{Given:} description of an NNCCM: integers $k$, $n$, $r$ and a sequence of~$r$~4-tuples in
       $\left( [1,k] \times [0,n] \times [1,k] \times [0,n] \right)^r $. \\
       \textbf{Parameter:} number of counters $k$. \\
       \textbf{Question:} is there an accepting run of the NNCCM?
    \end{verse}
    
    \begin{theorem}[Bodlaender, Groenland, Nederlof, Swennenhuis~\cite{BodlaenderGNS22a}]

        \textsc{NNCCM Acceptance} is complete for \textup{XNLP}.
    \end{theorem}

\subsection{Parameter Relations}\label{subsection:parameter_relations}
In this subsection, we discuss relations between the parameters considered in this paper.

    Based on the existing reductions in this paper and in the literature, there seems to be two families of parameters: one for width-like parameters, and one for bandwidth-like parameters. In the first family, we have the width for graphs and scheduling instances, and the number of source words in \textsc{Shuffle Product}. In the second family, we have the target value for bandwidth-like problems, the maximum delay value in scheduling with precedence delays, and the number of machines in the parallel identical machine setting. Despite the apparent correspondence between parameters in the same family, the relations between both parameter families seem to differ depending on the nature of the problem.

    Regarding \SchedulingBaseProblem{} and variants with non-unit processing times, w.l.o.g. the number of machines can be assumed to be at most the width. Indeed, in any schedule meeting all precedence constraints, there can never be more jobs concurrently scheduled than the width. As a result, e.g., the XNLP-completeness result with respect to the number of machines plus width in~\cite{BodlaenderGNS22a} actually showed XNLP-completeness parameterized by width.

    However, regarding scheduling with precedence delays, parameters width and maximum delay value~$\ell_{max}$ are typically incomparable. For instance, recall that \textsc{Directed Bandwidth} is equivalent to single machine scheduling of unit-jobs with equal-length maximum precedence delays, where the directed bandwidth target value corresponds to the (maximum) delay value plus one. For general directed acyclic graphs, we can have constant target value and unbounded width with caterpillars --- i.e., a tree with one path that contains all vertices of degree at least three. Conversely, consider a path plus a vertex~$v$ with an arc from the beginning of the path to~$v$ and another arc from~$v$ to the end of the path. Then the width is equal to two, while the target value could be made arbitrarily large by adjusting the length of the path.

    In contrast, for directed trees, the width is a strictly larger parameter than the directed bandwidth target value. Indeed, we propose a topological ordering of the vertices with directed bandwidth value at most $2 \cdot width - 1$. E.g., for downward directed trees, we label vertices by their depth starting from the root, then consider them by non-decreasing depth, and in any order at equal depth. By the definition of the width, there are at most $width$ vertices for each depth value, and the difference between the depth of neighbouring vertices is equal to one. Thus the distances between neighbouring vertices in the proposed topological order is at most~$2 \cdot width - 1$, and w.l.o.g. the target value can be assumed to be at most equal to this value.

\section{Shuffle Product and its Consequences}\label{section:shuffle_product}
    Given a word $s$, its length is denoted by $|s|$, and its $p$th letter in $s$ is denoted by~$s\IDLetter{p}$ ($1 \le p \le |s|$).

    \begin{definition}[Shuffle product]\label{definition:shuffle_product}
        A word $t$ is in the \emph{shuffle product} of words $s_1, \ldots, s_k$ if and only if $t$ can be obtained by interleaving the letters of $s_1, \ldots, s_k$, while keeping the order of the letters from the original words.
        
        Formally, we denote $s_1 \text{ $\shuffleProduct$} s_2 \text{ $\shuffleProduct$} \cdots \text{ $\shuffleProduct$} s_k$ the set of words $t$ for which there is be a bijective mapping $f: \left(\cup_{1 \le i \le k}(\{i\} \times [1, |s_i|])\right)  \rightarrow [1, |t|]$ such that, for all $i \in [1,k]$, sequence $(f(i,p))_{1 \le p \le |s_i|}$ is increasing and, for all $p \in [1, |s_i|]$,  $s_i\IDLetter{p} = t\IDLetter{f(i,p)}$.
    \end{definition}
    \begin{verse}
        \textsc{(Binary) Shuffle Product} \\
        \textbf{Given:} source words $s_1, \ldots, s_k$ and target word $t$ over a (binary) alphabet~$\Sigma$. \\
       \textbf{Parameter:} number of source words $k$. \\
       \textbf{Question:} is $t$ in $s_1 \shuffleProduct s_2 \shuffleProduct \cdots \shuffleProduct s_k$?
    \end{verse}\label{definition:shuffle_product_problem}
    
    Figure~\ref{fig:shuffle_product_example} gives an example. \textsc{Shuffle Product} is NP-hard for unbounded $k$ even over alphabets of size~3 (see Theorem 3.1 by Warmuth and Haussler~\cite{WarmuthHaussler:84:On-the-complexity}), while it is polynomial-time solvable for constant $k$ using dynamic programming~\cite{Mansfield:83:On-the-computational}. Rizzi and Vialette \cite{RizziVialette:13:On-Recognizing} asked about the parameterized
    complexity of \textsc{Shuffle Product}. This was partially answered by van Bevern et al., who showed the W[2]-hardness of \textsc{Binary Shuffle Product}~\cite{vanBevern16}. In the next subsection, we give a definite answer to this open question by showing XNLP-completeness of \textsc{(Binary) Shuffle Product}.

    \begin{figure}[tbp]
        \centering
        \begin{tikzpicture}
            \newcommand{\Vinit}{-0.035}
            \newcommand{\Hstep}{0.6}
            \newcommand{\Vstep}{-0.4}
            \node at (0,\Vinit) {$s_1=$};
            \node at (0,\Vinit+\Vstep) {$s_2=$};
            \node at (0,\Vinit+2*\Vstep) {$s_3=$};
            \node at (0.08,\Vinit+3*\Vstep) {$t=$};
            \node at (2*\Hstep,0) {$c$};
                \draw[->,dashed] (2*\Hstep,0.5*\Vstep) -- (2*\Hstep,2.5*\Vstep);
            \node at (3*\Hstep,0) {$b$};
                \draw[->,dashed] (3*\Hstep,0.5*\Vstep) -- (3*\Hstep,2.5*\Vstep);
            \node at (6*\Hstep,0) {$a$};
                \draw[->,dashed] (6*\Hstep,0.5*\Vstep) -- (6*\Hstep,2.5*\Vstep);
            \node at (7*\Hstep,0) {$a$};
                \draw[->,dashed] (7*\Hstep,0.5*\Vstep) -- (7*\Hstep,2.5*\Vstep);
            \node at (\Hstep,\Vstep) {$a$};
                \draw[->,dashed] (\Hstep,1.5*\Vstep) -- (\Hstep,2.5*\Vstep);
            \node at (4*\Hstep,\Vstep) {$b$};
                \draw[->,dashed] (4*\Hstep,1.5*\Vstep) -- (4*\Hstep,2.5*\Vstep);
            \node at (9*\Hstep,\Vstep) {$c$};
                \draw[->,dashed] (9*\Hstep,1.5*\Vstep) -- (9*\Hstep,2.5*\Vstep);
            \node at (5*\Hstep,2*\Vstep) {$c$};
                \draw[->,dashed] (5*\Hstep,2.4*\Vstep) -- (5*\Hstep,2.5*\Vstep);
            \node at (8*\Hstep,2*\Vstep) {$a$};
                \draw[->,dashed] (8*\Hstep,2.4*\Vstep) -- (8*\Hstep,2.5*\Vstep);
            \node at (\Hstep,\Vinit+3*\Vstep) {$a$};
            \node at (2*\Hstep,\Vinit+3*\Vstep) {$c$};
            \node at (3*\Hstep,\Vinit+3*\Vstep) {$b$};
            \node at (4*\Hstep,\Vinit+3*\Vstep) {$b$};
            \node at (5*\Hstep,\Vinit+3*\Vstep) {$c$};
            \node at (6*\Hstep,\Vinit+3*\Vstep) {$a$};
            \node at (7*\Hstep,\Vinit+3*\Vstep) {$a$};
            \node at (8*\Hstep,\Vinit+3*\Vstep) {$a$};
            \node at (9*\Hstep,\Vinit+3*\Vstep) {$c$};

        \end{tikzpicture}
        \caption{Example of a word $t=acbbcaaac$ in shuffle product $s_1 \text{ $\shuffleProduct$} s_2 \text{ $\shuffleProduct$} s_3$ with source words $s_1 = cbaa$, $s_2 = abc$ and $s_3 = ca$.}
        \label{fig:shuffle_product_example}
    \end{figure}   
    
    \subsection{XNLP-completeness of (Binary) Shuffle Product}\label{subsection:sp_XNLP_completeness}
    
    In this subsection, we show our main result regarding \textsc{Shuffle Product}:
    
    \begin{theorem}\label{theorem:shuffle_product}
        \textsc{(Binary) Shuffle Product} is \textup{XNLP}-complete.
    \end{theorem}

\begin{proof}
    XNLP-membership is straightforward: we propose a non-deterministic algorithm where we memorize the advancement on the target word and each of the source words. For each letter in $t$ in order, we guess which of the source words it corresponds to. Note that we only keep track of position indices in the words, and we do not need to write down any letter in the working space. So, the alphabet size does not appear in the working space, and the latter is in~$\mathcal{O}(k \log(|t|))$.
    
    To show XNLP-hardness, we adapt the W[2]-hardness proof by van Bevern et al.~\cite{vanBevern16}. We reduce from \textsc{NNCCM Acceptance}. We are given an instance~$\mathcal{I}$ with integers $k$, $n$, $r$ and a sequence of $r$ 4-tuples in $[1,k] \times 
    [0,n] \times [1,k] \times [0,n] $. Recall that $k$~is the number of counters and our parameter, $n$~is the maximum value of the counter, $r$~is the number of checks and each 4-tuple represents a check, done in the same order as in the sequence.
    For $j \in [1,r]$, if the $j$th~check is $(c_1,n_1,c_2,n_2)$ then, for all couples $(i,p) \in [1,k] \times [0,n]$, we define:
    
    \begin{equation}
        \bBlockLength{i,j,p} = \begin{cases}
            2 & \text{if } (i,n-p) = (c_1, n_1) \text{ or } (i,n-p) = (c_2, n_2), \\
            1 & \text{otherwise.}
        \end{cases}
    \end{equation}
    
    Given a sequence of words $(w_i)_{1 \le i \le q}$, we denote $(\prod_{i=1}^{q} w_i)$ the concatenated word $w_1\cdots w_q$. If $w_1 = w_2 = \cdots = w_q$, then this word is denoted by~$(w_1)^q$.
    Let $N = 2kn+1$. We propose the following instance~$\mathcal{I}'$ of \textsc{Binary Shuffle Product} with $k+3$~words over $\Sigma = \{a,b\}$:
    
    {\centering
    $\begin{array}{rclrcl}
            \displaystyle s_i & = & \prod_{j=1}^{r} \left(\left(\prod_{p=0}^{n}ab^{\bBlockLength{i,j,p}}  \right)^N\right),  & t & = &  \left(\left(a^kb^{k+2}\right)^n a^kb^{k+1} \right)^{Nr} \left(a^kb^{k+2}\right)^n, \\
            s_{k+1} & = & a^{|t|_a - \sum_{i=1}^{k} |s_i|_a}  \text{, and } & s_{k+2} & = & b^{|t|_b - \sum_{i=1}^{k} |s_i|_b}. \\
    \end{array}$
    }

    Similarly to van Bevern et al.~\cite{vanBevern16}, we define positions and long and short blocks the following way.

    \begin{definition}
        A \emph{block} (resp. a \emph{$c$-block}) in a word $s$ is defined as a maximal consecutive subword using only one letter (resp. only letter~$c \in \Sigma$). A block has position $i$ in $s$ if it is the ith successive block in~$s$.
        
        In $t$, we call $b$-blocks of length~$k+1$ \emph{short} and $b$-blocks of length~$k+2$ \emph{long}. For $i \in [1,k]$, in $s_i$, we call $b$-blocks of length~$1$ \emph{short} and $b$-blocks of length~$2$ \emph{long}.
    \end{definition}

    \begin{remark}\label{rem:nb_blocks}
        Note that $t$ has $2(n+1)Nr + 2n$~blocks and, for $i \in [1,k]$, $s_i$ has exactly $2n$~less blocks --- i.e., $2(n+1)Nr$~blocks. Essentially, we dedicate $2(n+1)N$~blocks to each of the $r$ checks successively.
    \end{remark}

    \begin{figure}[tbp]
        \centering
        \begin{tikzpicture}
            \tikzstyle{word}=[anchor=south west]
            \tikzstyle{grayfill}=[gray, opacity=0.2]
            
            \newcommand{\Hinit}{-0.4}
            \newcommand{\Vinit}{-0.035}
            \newcommand{\Hstep}{0.4}
            \newcommand{\Vstep}{-0.5}
            \node[word] at (\Hinit,\Vinit) {$s_1=$};
            \node[word] at (\Hinit,\Vinit+\Vstep) {$s_2=$};
            \node[word] at (\Hinit,\Vinit+2*\Vstep) {$s_3=$};
            \node[word] at (\Hinit+0.18,\Vinit+3*\Vstep) {$t=$};
            \node[word] at (9*\Hstep,1.5*\Vstep) {\scriptsize $(\ldots)$};
            \node[word] at (19*\Hstep,1.5*\Vstep) {\scriptsize $(\ldots)$};
            \fill[grayfill] (6*\Hstep,-\Vstep) rectangle (7*\Hstep+0.05,3*\Vstep);
            \fill[grayfill] (16*\Hstep,-\Vstep) rectangle (17*\Hstep+0.05,3*\Vstep);
            \node[word] at (\Hstep,0) {$a$};
            \node[word] at (2*\Hstep,0) {$b$};
            \node[word] at (5*\Hstep,0) {$a$};
            \node[word] at (6*\Hstep,0) {$b$};
            \node[word] at (7*\Hstep,0) {$a$};
            \node[word] at (8*\Hstep-0.04,0) {$bb$};
            \node[word] at (11*\Hstep,0) {$a$};
            \node[word] at (12*\Hstep,0) {$b$};
            \node[word] at (13*\Hstep,0) {$a$};
            \node[word] at (14*\Hstep,0) {$b$};
            \node[word] at (15*\Hstep,0) {$a$};
            \node[word] at (16*\Hstep-0.04,0) {$bb$};
            \node[word] at (17*\Hstep,0) {$a$};
            \node[word] at (18*\Hstep,0) {$b$};
            \node[word] at (21*\Hstep-0.04,0) {$bb$};
            \node[word] at (22*\Hstep,0) {$a$};
            \node[word] at (23*\Hstep,0) {$b$};
            \node[word] at (\Hstep,\Vstep) {$a$};
            \node[word] at (2*\Hstep,\Vstep) {$b$};
            \node[word] at (3*\Hstep,\Vstep) {$a$};
            \node[word] at (4*\Hstep,\Vstep) {$b$};
            \node[word] at (5*\Hstep,\Vstep) {$a$};
            \node[word] at (6*\Hstep-0.04,\Vstep) {$bb$};
            \node[word] at (7*\Hstep,\Vstep) {$a$};
            \node[word] at (8*\Hstep,\Vstep) {$b$};
            \node[word] at (11*\Hstep,\Vstep) {$a$};
            \node[word] at (12*\Hstep,\Vstep) {$b$};
            \node[word] at (13*\Hstep,\Vstep) {$a$};
            \node[word] at (14*\Hstep,\Vstep) {$b$};
            \node[word] at (15*\Hstep,\Vstep) {$a$};
            \node[word] at (16*\Hstep,\Vstep) {$b$};
            \node[word] at (17*\Hstep,\Vstep) {$a$};
            \node[word] at (18*\Hstep,\Vstep) {$b$};
            \node[word] at (21*\Hstep,\Vstep) {$b$};
            \node[word] at (\Hstep,2*\Vstep) {$a$};
            \node[word] at (2*\Hstep,2*\Vstep) {$b$};
            \node[word] at (3*\Hstep,2*\Vstep) {$a$};
            \node[word] at (4*\Hstep,2*\Vstep) {$b$};
            \node[word] at (5*\Hstep,2*\Vstep) {$a$};
            \node[word] at (6*\Hstep,2*\Vstep) {$b$};
            \node[word] at (7*\Hstep,2*\Vstep) {$a$};
            \node[word] at (8*\Hstep,2*\Vstep) {$b$};
            \node[word] at (11*\Hstep,2*\Vstep) {$a$};
            \node[word] at (14*\Hstep,2*\Vstep) {$b$};
            \node[word] at (15*\Hstep,2*\Vstep) {$a$};
            \node[word] at (16*\Hstep,2*\Vstep) {$b$};
            \node[word] at (17*\Hstep,2*\Vstep) {$a$};
            \node[word] at (18*\Hstep-0.04,2*\Vstep) {$bb$};
            \node[word] at (21*\Hstep,2*\Vstep) {$b$};
            \node[word] at (22*\Hstep,2*\Vstep) {$a$};
            \node[word] at (23*\Hstep-0.04,2*\Vstep) {$bb$};
            \node[word] at (\Hstep,\Vinit+3*\Vstep) {$a^3$};
            \node[word] at (2*\Hstep,\Vinit+3*\Vstep) {$b^5$};
            \node[word] at (3*\Hstep,\Vinit+3*\Vstep) {$a^3$};
            \node[word] at (4*\Hstep,\Vinit+3*\Vstep) {$b^5$};
            \node[word] at (5*\Hstep,\Vinit+3*\Vstep) {$a^3$};
            \node[word] at (6*\Hstep,\Vinit+3*\Vstep) {$b^4$};
            \node[word] at (7*\Hstep,\Vinit+3*\Vstep) {$a^3$};
            \node[word] at (8*\Hstep,\Vinit+3*\Vstep) {$b^5$};
            \node[word] at (11*\Hstep,\Vinit+3*\Vstep) {$a^3$};
            \node[word] at (12*\Hstep,\Vinit+3*\Vstep) {$b^5$};
            \node[word] at (13*\Hstep,\Vinit+3*\Vstep) {$a^3$};
            \node[word] at (14*\Hstep,\Vinit+3*\Vstep) {$b^5$};
            \node[word] at (15*\Hstep,\Vinit+3*\Vstep) {$a^3$};
            \node[word] at (16*\Hstep,\Vinit+3*\Vstep) {$b^4$};
            \node[word] at (17*\Hstep,\Vinit+3*\Vstep) {$a^3$};
            \node[word] at (18*\Hstep,\Vinit+3*\Vstep) {$b^5$};
            \node[word] at (21*\Hstep,\Vinit+3*\Vstep) {$b^4$};
            \node[word] at (22*\Hstep,\Vinit+3*\Vstep) {$a^3$};
            \node[word] at (23*\Hstep,\Vinit+3*\Vstep) {$b^5$};
            \node[word] at (24*\Hstep,\Vinit+3*\Vstep) {$a^3$};
            \node[word] at (25*\Hstep,\Vinit+3*\Vstep) {$b^5$};
        \end{tikzpicture}
        \caption{Illustration of the reduction with $k=3$ source words, $n=2$ and $r=2$ checks $(1,0,2,0)$ and $(1,1,3,0)$. In the proposed block mapping, gray areas represent counter checks, where counters have values $(1,0,0)$ and $(1,0,1)$ respectively.}
        \label{fig:shuffle_product_reduction}
    \end{figure}   

\begin{oframed}
    \textbf{Intuition.}
    Figure~\ref{fig:shuffle_product_reduction} illustrates the reduction.
    We track the offset in block positions between $t$ and every $s_i$. By Remark~\ref{rem:nb_blocks} and the alternation of $a$-blocks and $b$-clocks in all the words, for each $s_i$ these offsets are even integers in $\{0, \ldots, 2n\}$ and form a non-decreasing sequence. After a division by two, these values are meant to represent the sequence of values for counter $i$.
    
    Checks are performed in sequence, along $2(n+1)N$ consecutive blocks each. A check $(c_1, n_1, c_2, n_2)$ is represented by a short $b$-block in which no `block skip' happens --- i.e., every $s_i$ contributes at least a letter to this block. Since short $b$-blocks in~$t$ are of length~$k+1$, source words $s_{c_1}$ and $s_{c_2}$ cannot contribute both a long $b$-block of length~two to this short $b$-block. This reflects that we cannot have both $c_1 = n_1$ and $c_2 = n_2$ during the corresponding check of the counter machine.
    Note that the existence of such a short $b$-block for each check is ensured by setting $N$ to a sufficiently large value --- i.e., more than the total number of possible block skips which, by Remark~\ref{rem:nb_blocks}, is $2kn$.
\end{oframed}

\begin{lemma}\label{lemma:shuffle_product_XNLP_1}
        Let $t, s_1, \ldots, s_{k+2}$ be the words in instance~$\mathcal{I}'$. If there is an accepting run for instance~$\mathcal{I}$, then $t$ is in $s_1 \shuffleProduct s_2 \shuffleProduct \cdots \shuffleProduct s_{k+2}$.
    \end{lemma}

    \begin{proof}
        Suppose that instance~$\mathcal{I}$ has an accepting run. For $i \in [1,k]$ and $j \in [1,r]$, let $p_{i,j}$ be the value of counter~$i$ when the $j$th check is performed. We show that $t$ is in the shuffle product of $s_1, \ldots, s_{k+2}$. First, we provide a mapping of the blocks in each $s_i$ to the blocks in $t$. We encode the value of the counters with matching position offsets between these blocks. For $i \in [1,k]$ and $j \in [1,r]$, all~$2(n+1)N$~blocks associated to the $j$th check have a position offset of $2p_{i,j}$ between $s_i$ and~$t$. Formally, for $\pi \in [2(n+1)N(j-1), 2(n+1)Nj - 1]$, we map the block at position~$\pi$ in $s_i$ to the block at position~$\pi+ 2p_{i,j}$ in~$t$. Then, we consider the blocks in $t$ successively and assign letters the following way: we first assign the letters of the block from $s_1$ (if there is one) in order, then from $s_2$, and so on until $s_k$. If some letters remain unassigned in the block in $t$, we assign letters from $s_{k+1}$ or $s_{k+2}$ in order. Let $f_{\mathcal{I}}$ be the resulting letter assignment.

        We show that $f_{\mathcal{I}}$ is well defined and meets all requirements from Definition~\ref{definition:shuffle_product}. First, note that values~$p_{i,j}$ are based on an accepting run of instance~$\mathcal{I}$. As such, for $i \in [1,k]$, $(p_{i,j})_{1 \le j \le r}$ is non-decreasing and the sequence of block position assignments from $s_i$ to $t$ is increasing. Furthermore, among the blocks assigned to the same block in $t$, this implies that there is at most one coming from each of the words $s_1, \ldots, s_k$. Knowing this, it is clear that we never assign more letters than are available for each $a$-block and long~$b$-block in~$t$. E.g., with long~$b$-blocks, by the definition of values $\bBlockLength{i,j,p}$, for every check at most two of the source words $s_1, \ldots, s_k$ can assign a $b$-block with two letters, and all the others can only assign a $b$-block with one letter. Thus, at most $k+2$~letters are assigned to each long~$b$-block in~$t$.
        
        Now, when it comes to the short $b$-blocks in $t$, we show that at most one assigned block from $s_1, \ldots, s_k$ has two letters. This is ensured by the fact that values $p_{i,j}$ are based on an accepting run of instance~$\mathcal{I}$. Indeed, by contradiction, suppose that more than $k+1$ letters are assigned to a long $b$-block in~$t$ associated to the $j$th check $(c_1,n_1,c_2,n_2)$, $j \in [1,r]$. Then, by the definition of values $\bBlockLength{i,j,p}$, there are two $b$-blocks with two letters assigned to it from source words $s_{c_1}$ and~$s_{c_2}$, with a position offset of $2p_{c_1, j}$ and $2p_{c_2, j}$ respectively. Again, by the definition of values $\bBlockLength{i,j,p}$, we deduce that $p_{c_1, j} = n_1$ and $p_{c_2, j} = n_2$, which would mean that the accepting run fails on the $j$th check.

        Consequently, the proposed letter assignment $f_{\mathcal{I}}$ is well defined, injective and such that, for all $i \in [1,k]$, $(f_{\mathcal{I}}(i,p))_{1 \le p \le |s_i|}$ is increasing. Finally, by the definition of $s_{k+1}$ and $s_{k+2}$, we conclude that $f_{\mathcal{I}}$ is bijective, and thus that $t$ is in the shuffle product of $s_1, \ldots, s_{k+2}$. \qed
    \end{proof}

    \begin{lemma}\label{lemma:shuffle_product_XNLP_2}
        Let $t, s_1, \ldots, s_{k+2}$ be the words in instance~$\mathcal{I}'$. If $t$ is in the shuffle product of $s_1, \ldots, s_{k+2}$, then there is an accepting run for instance~$\mathcal{I}$. 
    \end{lemma}
    \begin{proof}
        Let $f$ be a bijective mapping from the letter positions of words $s_1, \ldots, s_{k+2}$ to the letter positions of word~$t$ accordingly to Definition~\ref{definition:shuffle_product}. Let $f'$ be the corresponding mapping which, for each couple $(i, \pi)$ representing a word $s_i$ with $i \in [1,k]$ and a block position in $s_i$, considers the first letter of the corresponding block, and returns the position of the block in~$t$ in which this letter is assigned by~$f$. Then, by the definition of~$f$ and the alternation of $a$-blocks and $b$-blocks in our words, for all $i \in [1,k]$, sequence $(f'(i, \pi))_{1 \le \pi \le 2(n+1)Nr}$ is increasing.

        We build an accepting run from the offsets of block positions between source words~$s_1, \ldots, s_k$ and target word~$t$. Given $i \in [1,k]$ and $\pi \in [1,2(n+1)Nr]$, let: 
        \begin{equation}
            \delta(i,\pi) = f'(i, \pi) - \pi
        \end{equation}
        By Remark~\ref{rem:nb_blocks} and the alternation of $a$-blocks and $b$-blocks in our words, we know that ${ \delta(i, \pi) \in \{0, 2, \ldots, 2n\} }$. Plus, by the previous paragraph, for all $i \in [1,k]$, $(\delta(i, \pi))_{1 \le \pi \le 2(n+1)Nr}$ is a non-decreasing sequence. As such, we intend to describe the values of each counter $i \in [1,k]$ in a run of instance~$\mathcal{I}$ with a subsequence of $(\delta(i, \pi)/2)_{1 \le \pi \le 2(n+1)Nr}$.

        In order to define such subsequences, we first define $r$ positions of short $b$-blocks in~$t$, denoted by $\pi_1, \ldots, \pi_r$, where no \emph{block skip} occurs --- i.e., an index $i \in [1,k]$ and a block position $\pi$ such that $\delta(i, \pi) < \delta(i, \pi+1)$. In other words, for all $j \in [1,r]$, for all $i \in [1,k]$, there must be $\pi^{(i,j)}$ such that ${ f'(i,\pi^{(i,j)}) = \pi_j }$ and $\delta(i,\pi^{(i,j)}) = \delta(i,\pi^{(i,j)} + 1)$.  Let $j \in [1,r]$. Consider the interval of block positions in~$t$: ${I_j = [2(n+1)(N(j-1) + 1), 2(n+1)Nj] }$. There are $N$ short $b$-blocks, and we claim that at least one of them is not involved in a block skip. Indeed, by the properties of $\delta(i, \pi)$ established in the previous paragraph, there are at most $2kn$ block skips which can occur in total. Since $N > 2kn$, by the pigeonhole principle, there is at least one short $b$-block in~$t$ with a position in~$I_j$ where no block skip occurs. We set $\pi_j$ as the position of one such short $b$-block. Then, for all couples $(i,j) \in [1,k] \times [1,r]$, we set $\pi^{(i,j)}$ as the unique block position in~$[1,2(n+1)Nr]$ in~$s_i$ such that $f'(i,\pi^{(i,j)}) = \pi_j$.

        Now, for instance~$\mathcal{I}$, we propose the run where, for all $(i,j) \in [1,k] \times [1,r]$, counter~$i$ has value $\delta(i, \pi^{(i,j)})/2$ during the $j$th check. We claim that this is an accepting run. Since we already established that this defines a non-decreasing sequence of values in $[0,n]$ for each counter, it remains to verify that all checks are met. Let $(c_1, n_1, c_2, n_2)$ be the $j$th check. By contradiction, suppose that $\delta(c_1, \pi^{(c_1,j)})/2 = n_1$ and $\delta(c_2, \pi^{(c_2,j)})/2 = n_2$. Then, by the definitions of the $\pi^{(i,j)}$ and $\bBlockLength{i,j,p}$ values, we have $f'(c_1,\pi^{(c_1,j)}) = f'(c_2,\pi^{(c_2,j)}) = \pi_j$, and both $b$-blocks, from $s_{c_1}$ and $s_{c_2}$ respectively, are long. Since no block skip occurs at position $\pi_j$, at least $k+2$ letters are assigned to the corresponding short $b$-block in~$t$, which leads to a contradiction. We conclude that all checks are met and the proposed run is indeed accepting for instance~$\mathcal{I}$.\qed
    \end{proof}

We now can conclude Theorem~\ref{theorem:shuffle_product}.
\end{proof}

\subsection{Consequences}\label{subsection:sp_consequences}

    In this subsection, from the XNLP-completeness of \textsc{Binary Shuffle Product}, we derive several implications to scheduling with precedence constraints.

    First, we consider \SchedulingBaseProblem{}, denoted by $P|prec, p_j=1|C_{max}$ in the three-field notation. This problem is NP-complete~\cite{Ullman:75:NP-complete}, and it was recently shown to be XNLP-complete parameterized by the number of machines~$m$ plus width~\cite{BodlaenderW20}. Although this implies XNLP-hardness parameterized by $m$, it is open whether the problem is NP-hard for constant~${m \ge 3}$~\cite{Ullman:75:NP-complete}.
    
    In contrast, when $p_j$ or $size_j$ are allowed to be non-unit, the problem is known to be NP-complete for a constant number of machines~\cite{BlazwiczLiu:96:Scheduling,Ullman:75:NP-complete}. In fact, van Bevern et al. showed that problems ${P3|prec, p_j=1, size_j\in\{1,2\}|C_{max} \le D}$ and ${P2|prec, p_j\in\{1,2\}|C_{max} \le D}$ are W[$2$]-hard parameterized by width (see Section~2.3 in~\cite{vanBevern16}). They first proved the W[$2$]-hardness of \textsc{Binary Shuffle Product}, then proposed reductions to both scheduling problems. It is not hard to see that these reductions only require logarithmic working space and thus can be reused as is to infer XNLP-hardness. Furthermore, XNLP-membership can be derived from straightforward adaptations of the algorithm proposed in the full version of~\cite[Section 4.3]{BodlaenderW20}.

    \begin{corollary}\label{corollary:consequence_parallel}
        Scheduling problems ${P3|prec, p_j=1, size_j\in\{1,2\}|C_{max} \le D}$ and ${P2|prec, p_j\in\{1,2\}|C_{max} \le D}$ are XNLP-complete parameterized by width.
    \end{corollary}

    Now, we consider scheduling with minimum precedence delays. On a single machine, the problem is already NP-complete with unit-time jobs and equal-length precedence delays --- which, in our extended three-field notation, can be denoted by $1|prec(\ell^{min}), \ell_{i,j} = \ell, p_j=1|C_{max} \le D$ ~\cite{LeungVornberger:84:On-some-variants} ---, albeit it can be solved in polynomial time if $G_{prec}$ is an upward (or downward) directed forest~\cite{BrunoJonesSo:80:Deterministic}.
    
    Regarding (maximum) delay value~$\ell$ as parameter, the situation is similar to the parallel identical machine parameterized by the number of machines. In~\cite{MallemHanenMunier:25:Single}, it was shown that $1|prec(\ell^{min}), \ell_{i,j} = \ell, p_j=1|C_{max} \le D$ is XNLP-complete parameterized by $\ell$ plus width. Although this implies XNLP-hardness parameterized by $\ell$, it is open whether the problem is NP-hard for constant~$\ell \ge 3$.
    However, when processing times can be equal to one or two, the problem is known to be NP-complete for constant $\ell \ge 2$~(see Section~8.2.4 in~\cite{Engels:00:Scheduling}); of note that the reduction required the width to be unbounded.

    From Corollary~\ref{corollary:consequence_parallel}, we can deduce that the restriction to constant~$\ell \ge 3$ and processing times in~$\{1,2\}$ is XNLP-complete parameterized by width. Indeed, starting from problem ${P3|prec, p_j=1, size_j\in\{1,2\}|C_{max} \le D}$, the reduction from \SchedulingBaseProblem{} proposed in~\cite[Section~5]{MallemHanenMunier:25:Single} can be easily adapted to the case of jobs with non-unit~$size_j$: we reduce in the same manner, i.e., every time unit in the parallel machine problem corresponds to a time interval of length~$m$ in our single machine instance with minimum precedence delays. Then, jobs of size~$size_j$ in the original instance are turned into jobs of processing time~$p'_j = size_j$. Plus the delay value is $\ell = m = 3$, and jobs of size~two are turned into jobs of processing time~two. This shows XNLP-hardness. Finally, note that XNLP-membership was already established in~\cite[Section~4]{MallemHanenMunier:25:Single}.
    
    \begin{corollary}\label{corollary:consequence_min_delays}
        Scheduling problem ${1|prec(\ell^{min}), \ell_{i,j}=3, p_j\in\{1,2\}|C_{max} \le D}$ is XNLP-complete parameterized by width.
    \end{corollary}
    
    This improves upon the results in~\cite{Engels:00:Scheduling,MallemHanenMunier:25:Single} about single machine scheduling with equal-length minimum precedence delays.

\section{Directed Bandwidth}\label{section:directed_bandwidth}
    
    This section shows that \textsc{Directed Bandwidth} is XNLP-complete. More precisely, we give
    two results: we show that \textsc{Directed Bandwidth} parameterized by the target bandwidth and the width is complete
    for XNLP, and that \textsc{Directed Bandwidth} parameterized by the target bandwidth is complete for XNLP when the input is a downward directed tree or an upward directed tree. 
    
    A directed graph $G=(V,A)$ is \emph{acyclic} (i.e., a \emph{DAG}), if it has no directed cycle.
    A \emph{topological ordering} of a directed acyclic graph $G=(V,A)$ with $|V|=n$ is a bijective function
    $f: V \rightarrow [0,n-1]$, such that for each arc $vw\in A$: $f(v)<f(w)$.
    The \emph{bandwidth} of a topological ordering $f$ is $\max_{vw\in A} f(w)-f(v)$. The
    \emph{directed bandwidth} of a directed acyclic graph $G$ is the minimum bandwidth
    over all topological orderings of $G$.
    The \emph{width} of a directed acyclic graph $G=(V,A)$ is the maximum size of a set $S\subseteq V$,
    such that there is no pair of vertices $v,w\in S$, $v\neq w$ with a directed path from $v$ to $w$ in $G$, i.e., it is the maximum size of an antichain in the partial order that is formed by the closure of $G$.
       
    In the \textsc{Directed Bandwidth} problem, given a directed acyclic graph $G=(V,A)$,
    and an integer $k$, we ask whether $G$ has directed bandwidth at most $k$. \label{definition:directed_bandwidth}
    The \textsc{Directed Bandwidth} problem is the natural directed variant of the well-studied
    \textsc{Bandwidth} problem, where we ask for an undirected graph $G=(V,E)$ and integer $k$,
    if there is a bijective mapping $f:V \rightarrow [0,n-1]$, with for each $\{v,w\}\in E$: $|f(v)-f(w)|\leq k$. 
    Already in 1978, Garey et al.~\cite{GareyGJK78} showed that \textsc{Directed Bandwidth} is
    NP-complete, for directed trees with no vertices of indegree more than two, and they
    showed that the problem is polynomial for $k=2$. In this paper, we focus on its parameterized complexity. 
    \textsc{Bandwidth} was shown first to be in XP by Saxe~\cite{Saxe80}, with a faster variant
    given by Gurari and Sudborough \cite{GurariS84}; these algorithms can be easily modified to give algorithms for \textsc{Directed Bandwidth} with the same running time. In~\cite{Bodlaender:21:Parameterized}, it was shown that \textsc{Bandwidth} and \textsc{Directed Bandwidth} are both W[$t$]-hard for all~$t$, even if the input graph (resp. the undirected version of the input graph) is a caterpillar --- i.e., a tree with one path that contains all vertices of degree at least three. Later, in~\cite{BodlaenderGNS22a}, it was shown that \textsc{Bandwidth} is complete for XNLP, again over caterpillars.
    
    As observed by Mallem~\cite[Section~3.1]{Mallem:24:Parameterized}, \textsc{Directed Bandwidth} is equivalent to scheduling jobs with unit times on one processor with equal-length maximum precedence delays - denoted by $1|prec(\ell^{max}), \ell_{i,j} = \ell, p_j = 1|C_{max} \le D$ in our extended three-field notation for scheduling problems.
    Mallem poses as an open problem whether \textsc{Directed Bandwidth} is
    XNLP-complete parameterized by the target bandwidth; this paper solves this open problem, even for directed trees.
    
    Below, in Section~\ref{subsection:dbdags}, we first give the complexity result for
    directed acyclic graphs with target bandwidth and width of the DAG as parameter. Then,
    in Section~\ref{subsection:dbtrees}, we show how we can modify this proof to show XNLP-hardness for \textsc{Directed Bandwidth} for trees.
    Our proofs are inspired by the XNLP-completeness proof for the undirected case from \cite{BodlaenderGNS22a}.

\subsection{Directed Bandwidth for Directed Acyclic Graphs}
\label{subsection:dbdags}

    In this subsection, we show the following result.
    
    \begin{theorem}
        \textsc{Directed Bandwidth} is \textup{XNLP}-complete for directed acyclic graphs, with as parameter target value plus width.
        \label{theorem:dbdag}
    \end{theorem}
    
    Membership can be shown the same way as membership for the undirected variant of the problem (e.g., see~\cite{BodlaenderGNS22a}): with target value $k$, repeatedly guess the next vertex in the ordering, and keep the last $k$ vertices in the ordering in memory.
    
\begin{oframed}\textbf{Intuition.}
    To show XNLP-hardness, we use a transformation from \textsc{NNCCM Acceptance}.
    Given an instance of \textsc{NNCCM Acceptance} with $k$~counters, we build a directed acyclic graph $G$ with width $k+O(1)$, such that $G$ has directed bandwidth at most $k+6$, if and only if the NNCCM has an accepting run. We have a properly chosen value $L$.
    $G$ has $(k+6)L+1$ vertices with
    exactly one vertex $v_{f,0}$ of indegree 0, and
    exactly one vertex $v_{f,L}$ of outdegree 0. $v_{f,0}$ must be the first, and $v_{f,L}$ the last vertex in the ordering. Then, we have a \emph{floor path} with length exactly $L$ from 
    $v_{f,0}$ to $v_{f,L}$, which means that every
    $(k+6)$th~vertex in the ordering must be a vertex in the floor path. Between each pair of successive vertices of this floor path, we thus have $L-1$ positions, which we call a \emph{batch}. 

    There are $k+1$ other paths from
    $v_{f,0}$ to $v_{f,L}$. $k$ of these paths represent a
    counter; each counter gadget path has length $n+L$. Between each pair of successive floor
    path vertices, we must have at least one vertex
    of each counter gadget path. The increase of a 
    counter is modelled by having more than one vertex of the counter gadget path between two
    successive floor path vertices --- thus a value of a counter is the difference between the number of steps of the counter gadget path and the floor path.

    We add additional vertices, which we call \emph{between} vertices. The main idea behind the construction is that, if the increase of counters models an accepting run of the NNCCM, then the batches have enough space to place all between vertices; but if the run rejects, then there is a batch modelling the check that halts the machine where we do have insufficient space to fit all between vertices.
    
    The last path is the \emph{filler path}: this path has precisely the length to ensure that the
    total number of vertices equals $(k+6)L+1$ --- this ensures that between each pair of successive floor path vertices, the distance is exactly~$L$.
    
    Figure~\ref{fig:db-dag} illustrates part of the construction.
\end{oframed}
    
    \begin{figure}
            \centering
            \includegraphics[width=1\linewidth]{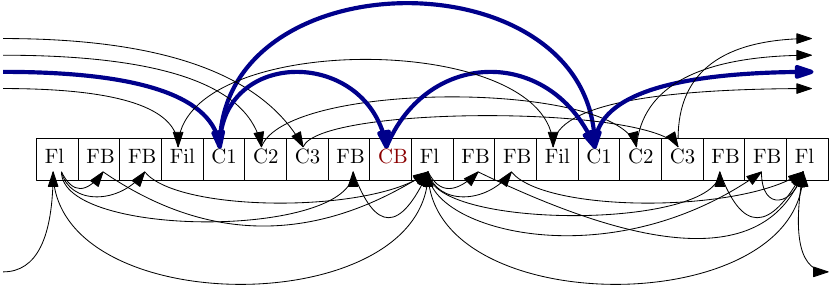}
            \caption{Part of the construction of $G$ for \textsc{Directed Bandwidth} for \emph{DAGs}. Fl = floor. Fil = filler path. C$i$ = Counter path $i$. FB = \emph{between} vertex of floor. CB =  \emph{between} vertex of counter gadget. The arcs of counter gadget 1 are shown in blue with fat lines. In the example, counter 1 is part of a check and has the value given in the check; thus counter path 1 has a \emph{between} vertex. 
            If another counter would also give the check value, there would be a second CB vertex which does not fit in between the two same consecutive floor vertices --- halting the machine corresponds to not being able to realize bandwidth $B$.}
            \label{fig:db-dag}
        \end{figure}

    \medskip
    
    To show XNLP-hardness, we reduce from \textsc{NNCCM Acceptance}. We first describe the reduction. Suppose we have an instance of \textsc{NNCCM Acceptance}, with $k$ counters that have values
    in $[0,n]$, and a series of $r$ checks. We will construct a directed acyclic graph $G$, and set a target value $B$, such that $G$ has directed bandwidth at most~$B$, if and only if the given NNCCM has at least one accepting run. 
    
    Set $B= k+6$, $S = (k+2)n$ and
    $N = (r+1)SB+1$. Let $L = (r+1)S$. The graph $G$ is constructed as follows.

    \paragraph{Floor.} We take a directed path with $L+1$ vertices, $v_{f,0}, \ldots, v_{f,L}$. We call this path the \emph{main floor path}. To this floor path we add \emph{floor between vertices}: for each pair of vertices $v_{f,i}, v_{f,i+1}$ we have from two to four vertices with an arc from $v_{f,i}$ and an arc to $v_{f,i+1}$:
        \begin{itemize}
            \item If $i$ is of the form $kni'$ for some $i'$, then we have three floor-between vertices between
            $v_{f,i}, v_{f,i+1}$. The third we call the \emph{extra floor-between} vertex.
            \item If $i$ is of the form $S\cdot i+1$ for some $i$, then we have four floor-between vertices between
            $v_{f,i}, v_{f,i+1}$. The third and fourth we call the \emph{extra floor-between} vertices.
            \item Otherwise, there are two floor-between vertices.
        \end{itemize}
    
    \paragraph{Counter gadgets.}
    For each of the $k$ counters, we have a \emph{counter gadget}. The counter gadget has a similar shape at the
    floor path: a long path with a number of between-vertices. The path is slightly longer  (precisely by $n$), and only a few pairs have a between-vertex.
    
    A counter gadget has a \emph{counter gadget path} with
    $L+n$ vertices, which we name $v_{c,i,j}$, with $i\in [1,k]$, and
    $j\in [1, L+n]$. There is an arc from $v_{f,0}$ to the first vertex of each counter gadget path, and an arc from the last vertex of each counter gadget path to
    $v_{f,L}$.
    
    If the $\alpha$th check of the NNCCM is the four-tuple $(i_1,q_1, i_2, q_2)$, then
    we add a between vertex that has an arc from $v_{c,i_1 \alpha\cdot S  + n - q_1}$ 
    and an arc to $v_{c,i_1, \alpha\cdot S  + n - q_1} +1$; and we also
    add a between vertex that has an arc from $v_{c,i_2, \alpha\cdot S  + n - q_1}$
    and an arc to $v_{c,i_2, \alpha\cdot S  + n - q_1} +1$.
    
    \paragraph{Filler path.}
    The last step in the construction of $G$ is the \emph{filler path}. We count the number
    of vertices in the floor and counter gadgets. Suppose this is $\alpha$.
    Let $Q= N - \alpha = L+1 - \alpha $.
    Then,
    take a directed path with $Q$ vertices. We take an arc
    from $v_{f,0}$ to the first vertex of of the filler path, and an arc from the last vertex of the filler path to
    $v_{f,L}$.
    
    This finishes the construction of $G$. 
    
    \paragraph{Correctness of the reduction.}
    We start with a few observations. First, notice that $G$ has one vertex with indegree $0$, namely $v_{f,0}$ and one vertex with outdegree zero, namely $v_{f,L}$.
    Second, notice that the total number of vertices equals $N$. 
    
    Also, note that the width of $G$ is $k+5$: a set $Q$ of vertices without a directed path between any pair of them can contain at most four vertices of the floor gadget, one vertex per counter gadget, and one vertex of the filler path.
    
    \begin{lemma}
        If the \textup{NNCCM} has an accepting run,
        then $G$ has a topological ordering of bandwidth at most $B$.
        \label{lemma:dirbandw1-1}
    \end{lemma}
    
    \begin{proof}
        First, map each vertex of the floor path: $f(v_{f,i})= B \cdot i$.
        
        Now, we first map the vertices of the counter paths. 
        
        We take a variant of the NNCCM model. Again, we start with $k$ counters that are initially $0$. We have $L$ time steps, where we repeat the following $r$ times: first $n$ time steps, nothing happens, then $kn$ time steps, we either do nothing or non-deterministically increase one of the counters by one (but not to a value larger than $n$), then we do nothing for $n-1$ steps, and finally we do the next of the $r$ checks. After the $r$th check, we first do nothing for $n$ steps, then each step, when there is a counter smaller than $n$, we increase one counter by one, till all counters have value $n$, after which we do nothing till we reach the $L$th step.
        
        It is easy to see that this gives an alternative, equivalent model of the NNCCM. So, we assume we have an accepting run of the NNCCM in this alternative model.
        
        We denote \emph{the $i$th batch} the vertices between $v_{f,i-1}$ and $v_{f,i}$ in the ordering that we are constructing. Note that each batch contains $B-1 = k+5$ vertices.
        These $k+5$ positions are distinguished as follows:
        \begin{itemize}
            \item The first two positions after $f(v_{f,i-1})$ are the \emph{early star} positions. These will be used for between vertices or vertices of the filler path, depending on a case analysis, which is explained below.
            \item The third position is reserved for the filler path.
            \item Then, we have $k$ positions, each reserved for one counter gadget path.
            \item The last two positions in the batch, the two positions before $f(v_{f,i})$ are called the \emph{late star} positions. Depending on a case analysis, these are used for
            between vertices, the filler path, or a vertex from a counter path --- the latter models the increase of a counter by one.
        \end{itemize}
        
        \begin{figure}
            \centering
            \includegraphics[width=0.5\linewidth]{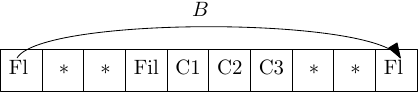}
            \caption{Illustration to the proof of Lemma~\ref{lemma:dirbandw1-1}, with a part of the layout from one floor path vertex to the next. Fl = position of floor path vertex. Fil = position reserved for the filler path. Ci = position reserved for the counter gadget path of the $i$th counter. * = position used by between vertices, counter path vertices and filler path vertices.}
            \label{fig:directed-one}
        \end{figure}
        
        We first position the vertices of the counter gadget paths. Consider the $i$th counter at
        time step $j$. If this counter is not increased and has value $\alpha$ at time step $j$,
        then we put the $(j+\alpha)$th vertex of the counter gadget path for the $i$th counter
        at position $f(v_{f,i-1}) + 3+i$. If this counter is increased from $\alpha$ to $\alpha+1$,
        then we again put the $(j+\alpha)$th vertex of the counter gadget path for the $i$th counter
        at position $f(v_{f,i-1}) + 3+i$, and we put the $(j+\alpha +1)$th vertex of the counter gadget path for the $i$th counter 
        at position $f(v_{f,i})-2$ (i.e., the first of the two late star vertices.)
        
        In the next step of the construction, we place the between vertices of the counter gadgets.
        Recall that such a vertex has an arc from a vertex of the form
        $v_{c,i_1, \alpha\cdot S  + n - q_1}$, with $c\in [1,k]$ denoting a counter, and
        $\alpha\in [1,r]$ denoting a check.
        There are two cases:
        \begin{itemize}
            \item $v_{c,i_1 \alpha\cdot S  + n - q_1}$ is placed in a the $j$th slot with $j$ of
            the form $\alpha \cdot S+1$. Then, we place this counter gadget between vertex in one of the two early star positions after $v_{f,j+1}$. (I.e., these go to the next batch.)
            \item Otherwise, the counter gadget between vertex is placed in one of the two late
            star positions of the same batch as $v_{c,i_1 \alpha\cdot S  + n - q_1}$.
        \end{itemize}
        
        As only two counters are involved in a check, at each batch, there are at most two counter gadget between vertex, so we can place all these counter gadget between vertices.
        
        The crux of the proof is that we have sufficient space to place the between vertices of the floor. There are the following cases:
        \begin{itemize}
            \item First we look at the batch of a time step of the form $\alpha \cdot S$, $\alpha\in [1,r]$. Here we use that we have an accepting run. If there is a counter gadget with a between vertex
            to be placed in this batch, then we can observe that this implies that this counter has
            the value that makes one half of the $\alpha$th check succeed. As we have an accepting run, there is only of the two counters that has the value that gives a counter gadget between vertex in this batch. So, there are three yet unused star positions, and we can place the three between vertices in these.
            \item Next, we consider the batch of a time step of the form $\alpha \cdot S+1$, $\alpha\in [1,r]$. Here, we do not place (by construction) counter gadget between vertices, and we do not increase a counter. So, we have space for all four floor between vertices.
            \item Next, we look at a batch of a time step in $[\alpha \cdot S -n+1, alpha \cdot S+1]$ that is not one of the cases above. Here we have two floor between vertices, at most two
            counter gadget between vertices, and we do not increase a counter. So, these between vertices can be placed in the four star positions.
            \item At each other time step, we have no counter gadget between vertices, two floor between vertices, and possibly one extra counter gadget path vertex (for a counter that is increased by one, cf. the construction above.) We have enough star positions for these (actually, one or two are unused yet, and will be filled, see below, by the filler path.)
        \end{itemize}
        
        At this point, all we need to do is to place the vertices of the filler path. Notice that the positions reserved for the filler path are not yet assigned, and each has distance $B$ to the next such vertex. So, we can place the filler vertices in sequence on the yet unused positions.
        
        This completes the construction of the topological ordering with bandwidth at most $B$. \qed
    \end{proof}
    
    \begin{lemma}
        If $G$ has a topological ordering of bandwidth at most $B$, then the \textup{NNCCM} has an accepting run.
        \label{lemma:dirbandw1-2}
    \end{lemma}
    
    \begin{proof}
        For each of notation, we assume that we represent the topological ordering with a bijective mapping $f: V \rightarrow [0,N-1]$.
            First, as the only vertex with indegree is $v_{f,0}$, so 
        $f(v_{f,0}) = 0$. Similarly, as $v_{f,(r+1)kn}$ is the only vertex of
        outdegree 0, $f(v_{f,(r+1)kn}) = N-1 = (r+1)knB$. 
        Now, observe that the number of arcs from $v_{f,0}$ to $v_{f,(r+1)kn}$ is $(r+1)kn$, while the difference in position in $f$ is exactly $B$ times as large --- this implies that for each $i$,
        $f(v_{f,i}) = B \cdot i$.
        
        Call the set of vertices between two successive vertices of the floor path (again, like in the proof of the previous lemma) a \emph{batch}. Each batch has exactly $B-1 = k+5$ vertices.
        
        As the filler path forms a path from $v_{f,0}$ to $v_{f,L}$,
        each $B$ successive positions must contain at least one vertex
        from the filler path. Similarly, each $B$ successive positions
        contains at least one vertex from each counter gadget path.
        This implies that each batch contains at most four between
        vertices.
        
        For each of the $L$ batches, we assign a value to each counter.
        We say that the $i$th counter has \emph{value}
        $\beta$ at time $j$,
        if $v_{c,i,j+\beta}$ is in the $j$th batch. In case there are
        more vertices of the $i$th counter gadget path in this batch, we choose the highest value.
        
        First, notice that at time step one, each counter has value at least $0$. Second, because each batch contains at least
        one vertex of each counter gadget path, when we go to a next batch, we have at least the next vertex of the counter gadget path, so the value of a counter does not decrease with time.
        Third, as the last vertex of a counter gadget path is of the form
        $v_{c,i,L+n}$ and is in the $L$th batch, we have that at time $L$ the value of each counter equals $n$. So, counters have
        values between $0$ and $n$, do not decrease over time.
        
        Thus, these values of the counters can be used to describe
        the non-deterministic increasing of an NNCCM. It remains to
        verify that this behaviour does not let the machine halt. So, we end the proof by showing that
        at each time step of the form $\alpha \cdot S$, the values of
        the counters do not halt the machine.
        
        Suppose the $\alpha$th check is the four tuple $(c_1,n_1,c_2,n_2)$. Suppose this check halts the machine.
        Then counter $c_1$ equals $n_1$ at time
        $\alpha \cdot S$, hence 
        $v_{c,c_1,\alpha \cdot S + n_1}$ is in batch $\alpha \cdot S$.
        This vertex has an arc to a counter gadget between vertex.
        This between vertex must be in batch $\alpha \cdot S$ or in
        batch $\alpha \cdot S+1$. A similar argument shows that
        a counter gadget between vertex of the gadget for counter $c_2$ must be in batch $\alpha \cdot S$ or in
        batch $\alpha \cdot S+1$. 
        
        Now, we have in batch $\alpha \cdot S$ and
        batch $\alpha \cdot S+1$ in total at least nine between vertices: three floor between vertices for $\alpha \cdot S$,
        four floor between vertices for $\alpha \cdot S+1$ and
        two counter gadget between vertices. But, we have only space for at most eight between vertices.
        (The two slots, not counting floor path vertices, must contain each one filler vertex
        and $k$ counter gadget path vertices, which leaves space for $2(k+5)-2k-2 = 8$ between vertices.)
        This contradiction shows that the test cannot cause the machine to halt, and thus we have
        an accepting run of the NNCCM. \qed
    \end{proof}
        
    The result now follows from Lemma~\ref{lemma:dirbandw1-1} and Lemma~\ref{lemma:dirbandw1-2}, and the observations that $G$ has bounded width and can be constructed in logarithmic space.

\subsection{XNLP-hardness for Trees}
\label{subsection:dbtrees}

    In this subsection, we show the following result.
    
    \begin{theorem}
        \textsc{Directed Bandwidth} is \textup{XNLP}-complete, even for inputs that are either a downward directed tree or an upward directed tree, with as parameter target value.
            \label{theorem:db-tree}
    \end{theorem}
    
    Membership follows the same way as the XNLP-membership for DAGs. We now show XNLP-hardness for downward directed trees, and the result for upward directed trees follows directly (just reverse the direction of all arcs).

\begin{oframed}
    \textbf{Intuition.} The graph created in the proof of Theorem~\ref{theorem:dbdag} can be turned into an downward directed tree by removing the following arcs: for each \emph{between} vertex, its outgoing arc, and for each counter gadget path and the filler path, the arc to~$v_{f,L}$. 
    The outgoing arcs of \emph{between} vertices are actually not needed to make the reduction work,
    but keep the width of directed acyclic graph small --- without these, the width would no
    longer be bounded by a function of $k$. As the width of the
    tree is not a parameter in Theorem~\ref{theorem:db-tree}, we can thus afford to remove these arcs. 
    
    The other type of arcs however disables the proof that a topological ordering of bandwidth at most $B$ implies an 
    accepting run of the NNCCM --- in particular, we cannot guarantee that the last vertex of
    the floor path $v_{f,L}$ is the last vertex of the ordering, and we cannot guarantee that the last vertex of each counter gadget path and the filler path is in the last batch, i.e.,
    between $v_{f,L-1}$ and $v_{f,L}$ in the ordering. So, to ensure these, what we do below is to add
    new gadgets: for each of the floor, filler path, and counter gadgets, we add \emph{tails} --- a new subtree whose root has as parent the last vertex
    of the floor path, filler path, or a counter gadget path. By carefully choosing the number of leaves for each these tails, we can enforce that they are all at the end of the ordering and in a specific order --- the more leaves they have, the closer they are to the last position.
\end{oframed}
    
    \paragraph{The construction.}
    We build a graph $G''$ as follows. First, take the graph $G$, as constructed in the
    proof of Theorem~\ref{theorem:dbdag}, but omit the outgoing arcs from between vertices,
    and the arcs to $v_{f,L}$ from the last vertices of the filler path and the counter gadget paths. Call this downward directed tree $G'$. 
    
    Then, we do the following. Each of the floor path, filler path and counter gadget paths is made longer, and we add a number of pendant leaves to the last or last two vertices of the path. We call the new parts of these paths the \emph{tail paths} and the new pendant vertices the \emph{tail leaves}.
    
    \begin{itemize}
        \item Add $4$ new vertices that have $v_{f,L}$, the last vertex of the floor path as parent.
        \item The \emph{floor tail path} has $3k+3$ vertices, $t_{f,1}, \ldots, t_{f,3k+3}$; there is an arc from the last vertex of the floor path $v_{f,L}$ to the first vertex
        of the floor tail path~$t_{f,1}$. The last vertex of the floor tail path $t_{f,2k+2}$ gets $B$ tail leaves.
        \item The \emph{filler tail path} has $3k+3$ vertices, $t_{fi,1}, \ldots, t_{fi,3k+3}$; there is an arc from the last vertex of the filler path to the first vertex
        of the filler tail path. Both the one-but-last and the last vertex of the floor tail path get $B-2$ tail leaves.
        \item For each $i\in [1,k]$, we have the \emph{counter gadget tail path} for the $i$th counter. It has $3k+4-3i$ vertices; call these $t_{c,i,1}, \ldots, t_{c,i,3k+4-3i}$; there is an arc from the last vertex of the $i$th counter gadget path to the first vertex
        of the $i$th counter gadget tail path. To the last three vertices of each counter gadget tail path, i.e., to $t_{c,i,3k+2-3i}, t_{c,i,3k+3-3i}, t_{c,i,3k+4-3i}$ we add
        $B-2-i$ tail leaves.
    \end{itemize}

A simplified example is given in Figure~\ref{figure:tail}.

\begin{figure}
    \centering
    \includegraphics[width=0.6\linewidth]{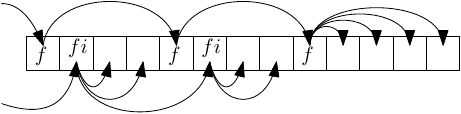}
    \caption{The last part of the tails. In the figure, the last three vertices of the floor tail path ($f$ --- arcs at upper side of image), the last two vertices of the filler tail path ($fi$~---~arcs at lower side of image), and their tail leaves are shown, with $B=4$. }
    \label{figure:tail}
\end{figure}
    
    This finishes the construction of $G''$. It is clear that $G''$ is a downward directed tree, and can be constructed in logspace.
    
    \begin{lemma}
        If the \textup{NNCCM} has an accepting run, the $G''$ has a topological ordering of directed bandwidth at most $B$.
    \end{lemma}
    
    \begin{proof}
        First, take the ordering of the vertices in $G'$ (i.e., in $G$) as in the proof of the first lemma in the proof of Theorem~\ref{theorem:dbdag}.
    
        First, we place the vertices of the floor tail path. We put consecutive ones at distance~$B$, like in
        the proof of Theorem~\ref{theorem:dbdag}: the $j$th vertex of the floor tail path is
        placed at position $f(t_{t,i}) = f(v_{f,L})+ B \cdot i$. Like before, this creates batches; we have
        tail batches $1, 2, \ldots, 3k+3$, with tail batch 1 formed by the positions between
        $v_{f,L}$ and $t_{f,1}$; for $i\in [2,3k+3]$, tail batch $i$ is formed by the positions
        between $t_{f,i-1}$ and $t_{f,i}$.
        
        Now, in tail batch $i$, we place successively: the $i$th vertex of the filler tail path,
        and for all $j \in [1,k]$, if the $j$th counter gadget tail path has length at least $i$,
        the $i$th vertex of the $j$-th counter gadget tail path. Then, in batch $1$, we place the
        four new leaves with $v_{f,L}$ as parent. 
        Then, this is followed by the tail leaves of the last placed filler or tail vertex. In each case, we precisely placed $B-1$ vertices in the batch.
        
        We finish the construction by placing the tail leaves of the floor at positions
        $f(t_{f,3k+3})+1$ till $f(t_{f,3k+3})+B$.
            
        A simple case analysis shows that this gives a topological ordering of directed bandwidth $B$.
        (For instance, the first vertices of filler tail and counter gadget tail paths have
        distance $B-2$ to the last vertex of their predecessors on the filler and counter gadget paths, cf. the construction in the proof of Theorem~\ref{theorem:dbdag}.) \qed
    \end{proof}
    
    \begin{lemma}
        If $G''$ has a topological ordering of directed bandwidth at most $B$, then the \textup{NNCCM} has an accepting run.
    \end{lemma}
    
    \begin{proof}
        Let $f$ be a topological ordering of $G''$
        of directed bandwidth at most $B$.
        
        First, notice as $t_{f,3k+3}$ has $B$ arcs to pendant leaves, the topological ordering must
        end with $t_{f,3k+3}$ and then these leaves --- this construction forms an `impassable barrier'. As the total number of vertices in all tails (including the four new successors of $v_{f,L}$) equals $(3k+4)B$, it follows
        that $f(t_{f,3k+3}) = N + (3k+3)B$. From this, we have that for all $i \in [1,3k+3]$,
        $f(t_{f,i}) = N + i\cdot B$. So, in this way, we created $3k+3$ floor batches, compare with the
        proof of Theorem~\ref{theorem:dbdag}; each batch the vertices between two successive
        vertices of the floor tail path, with the first floor batch between $v_{t,L}$ and
        $t_{f,1}$.
        
        Now, we claim that the last two batches must be filled by the vertices of the filler tail.
        Suppose we place a vertex $x$ of a counter gadget tail in floor batch~$3k+3$. Then, as there is
        a path from $v_{t,0}$ to $x$, each $B$ consecutive positions in the ordering contain at least one floor path or floor tail path vertex, and one vertex of this counter gadget path or its tail. But, the one-but-last vertex of the filler tail path  $t_{fi,3k+2}$ has $B-1$ successors, and thus one of these would be placed at a distance more than $B$ apart, contradiction.
        It follows that $t_{fi,3k+2}$ is placed directly after $t_{f,3k+1}$ and 
         $t_{fi,3k+3}$ is placed directly after $t_{f,3k+2}$, with the remaining positions in floor
         batches $3k+1$ and $3k+2$ used by the tail leaves of the filler path.
        
        \begin{claim}
            For each $i\in [1,k]$, $f(t_{c,i,3k+2-3i}) > f(v_{f,L})$.
        \end{claim}
        
        \begin{proof}
            Note that in other words, the claim states that for each counter gadget tail path, the second-but-last vertex is in the ordering after the last vertex of the floor path.
        
            Suppose not. Take the smallest value of $i$ for which the claim does not hold.
            Consider the $3B$ positions starting with $f(t_{c,i,3k+2-3i})$. These contain at least three
            vertices of the floor or floor tail; at least three vertices of the filler path or filler tail, and for each $j<i$, at least 
            three vertices of the $j$th counter gadget or counter gadget tail. (The latter follows as $f(t_{c,j,3k+2-3j}) > f(v_{f,L})$, hence
            every $B$ positions before $f(t_{c,j,3k+2-3j})$ must contain at least one vertex of the $j$th counter gadget or counter gadget tail.)
            Then, also, case analysis shows that this set contains at least four between vertices or the four successors of $v_{f,L}$. 
            In addition, all leaf vertices of the $j$th counter gadget tail
            are in this set, and the vertices $f(t_{c,i,3k+2-3i}). f(t_{c,i,3k+3-3i}), f(t_{c,i,3k+4-3i})$.
            This gives more than $3B$ vertices in $3B$ positions, contradiction. \qed
        \end{proof}
        
        \begin{claim}
            For all $1\leq i<i'\in [1,k]$, the last three vertices of the $i'$th counter gadget tail path are before each of the last three vertices of the tail leaf path of counter gadget $i$.
        \end{claim}
        
        \begin{proof}
            We show with induction to $i$ that this holds for all $i''\leq i$. The argument follows from counting vertices, and using that  each $B$ positions before the last vertex of a path, contain at least one vertex from a path. \qed
        \end{proof}
        
        Now, the only way that we can fit all the tails is that for each $i$, we place
        $t_{c,i,k+2-3i}$ in floor slot $t_{c,i,k+2-3i}$,
        $t_{c,i,k+3-3i}$ in floor slot $t_{c,i,k+3-3i}$, and
        $t_{c,i,k+4-3i}$ in floor slot $t_{c,i,k+4-3i}$. From this, it follows that tail slot 1
        contains all vertices~$t_{c,i,1}$.
        
        Now, we can show the lemma. Let $G'''$ be the graph $G$, obtained from removing all outgoing arcs that leave between vertices.
        Consider the restriction $f'$ of $f$ to the vertices in $G'$. 
        Suppose $x$ is either the last vertex of a counter gadget path or the last vertex
        of the filler path. $x$ is in $G''$ has an arc to the first vertex of the corresponding tail path --- that vertex is in the first tail slot, so $x$ is in slot~$L$. Hence,
        we have $f'(x) \leq f'(v_{f,L}) < f'(x)+B$. So, $f'$ is a topological ordering of~$G'''$
        of bandwidth at most $B$.
        
        Now, we can follow the proof of Theorem~\ref{theorem:dbdag} and show that the NNCCM has an
        accepting run, as that proof does not use the outgoing arcs from between vertices. (The only role of these arcs is to keep the width of the DAG small.) \qed
    \end{proof}
    
    From the two lemmas above, and the fact that $G''$ can be constructed in logspace, Theorem~\ref{theorem:db-tree}
    follows.
\section{Discussion}
    
    We end the paper with a few open problems.
    
    First, regarding scheduling with minimum delays, current reduction techniques from \textsc{Shuffle Product} to scheduling problems require jobs of either non-unit processing time or non-unit size. As such, it is still open whether $1|prec(\ell^{min}), \ell_{i,j} = \ell, p_j=1|C_{max} \le D$ is NP-complete for constant $\ell \ge 3$. In fact, following the relation with the parallel identical machine setting discussed in Section~\ref{subsection:sp_consequences}, this mirrors the longstanding open question about whether \SchedulingBaseProblem{} is NP-complete for a constant number of machines~$m \ge 3$~\cite{Ullman:75:NP-complete}.
    
    Second, regarding scheduling with maximum delays, settling the case of directed trees parameterized by width would complement both results about \textsc{Directed Bandwidth} in this paper. In this context, the parameter bounds both the number of leaves and the number of branching internal nodes in the tree. We believe that a fixed-parameter algorithm could stem from enumerating all sequences of these nodes which are compatible with the bandwidth constraints then, for each one of these sequences, solve an integer linear program featuring variables for the positions of these nodes in the topological order. This would imply fixed-parameter tractability, as long as the number of variables in each integer linear program is bounded by a function of the width~\cite{Lenstra:83:Integer}.
    
    Additionally, in 1978 Garey et al.~\cite{GareyGJK78} showed that \textsc{Directed Bandwidth} is NP-complete for directed trees with each vertex indegree at most two. In this paper, we only showed XNLP-completeness for general directed trees; we conjecture that this result could be strengthened to directed trees with indegree~two, but expect that a proof of such a result would be very technical.
    In~\cite{Bodlaender:21:Parameterized,BodlaenderGJJL22}, it was shown that \textsc{Bandwidth} is XNLP-complete for caterpillars with hair length at most three. We do not expect this result to hold for the directed variant --- i.e., if we have a downward or upward directed tree with all vertices of degree more than two on one path --- but conjecture this case to be polynomial time solvable.
    
    Finally, regarding single machine scheduling with exact delays, results seem to differ significantly between maximum delay value taken alone and combined with the width as parameter (see Tables 1 and 2 in~\cite{MallemHanenMunier:25:Single}). In particular, on a single machine with chains of unit-time jobs and equal-length delays, the problem is W[1]-hard with the former parameter and fixed-parameter tractable with the latter one. Crucially, note that this problem parameterized by maximum delay value is equivalent to~\textsc{Unary Bin Packing} (see Appendix~\ref{appendix:unary_bin_packing}). As such, improving upon the W[1]-hardness of scheduling with exact delays would also help settling the open question in~\cite{JansenKratschMarx:13:Bin}. So far, to the best of our knowledge, only XP-membership is known (e.g., a polynomial-time algorithm for a constant number of bins is mentioned in~\cite{GareyJohnson:79:Computers} and detailed in~\cite[Section~5.2]{GurskiRehsRethmann:19:Knapsack}). In contrast, slight generalizations, such as having either $\ell_{i,j} \in \{0, \ell\}$ or $p_j \in \{1,2\}$, are known to be NP-complete even when the maximum delay value is equal to one~\cite{MallemHanenMunier:25:Single}.





%
%
%
\bibliographystyle{splncs04}
\bibliography{papers}

\appendix

\section{Unary Bin Packing}\label{appendix:unary_bin_packing}

    In this appendix, we show the equivalence between \textsc{Unary Bin Packing} and single machine scheduling of (unary) chains of unit-time jobs with equal-length exact delays. In our extended three-field notation, chain-like precedence constraints are denoted by `\emph{chains}' instead of `\emph{prec}'.
    \begin{verse}
        \textsc{Unary Bin Packing} \\
        \textbf{Given:} Positive integers $B, k, a_1, \ldots, a_n$ all given in unary, respectively representing a bin capacity, a number of bins and $n$~items. \\
       \textbf{Parameter:} number of bins $k$. \\
       \textbf{Question:} is there a partition $(P_i)_{1 \le i \le k}$ of the items such that, for all $i \in [1,k]$, $\left(\sum_{a_j \in P_i} a_j\right) \le B$?
    \end{verse}

    \begin{theorem}
        Problem \textsc{Unary Bin Packing} is equivalent to scheduling problem ${\mathcal{P}_{ex} = 1|chains(\ell^{ex}), \ell_{i,j} = \ell, p_j=1|C_{max} \le D}$ parameterized by (maximum) delay value~$\ell$.
    \end{theorem}
    \begin{proof}
        A reduction from \textsc{Unary Bin Packing} to~$\mathcal{P}_{ex}$ was already proposed in~${\text{\cite[Section~3]{MallemHanenMunier:25:Single}}}$. The main argument was that, if a job in a chain was scheduled at a time~$t$, then all jobs in the chain must start at time~$t$ plus some multiple of~$(\ell+1)$. As a result, a bin could be represented a set of starting times with the same value modulo~$(\ell+1)$, with~$(\ell+1)$ as the number of bins.

        With this correspondence in mind, we propose the following reduction from $\mathcal{P}_{ex}$ to~\textsc{Unary Bin Packing}. Let $\mathcal{I'}$ be an instance of $\mathcal{P}_{ex}$. Let $B' = \lceil D/(\ell+1) \rceil$ and $r' = B' (\ell+1) - D$. Then $\mathcal{I'}$ can be seen as fitting unary items into $r'$~bins of capacity~$B'-1$ and $(\ell+1-r')$ bins of capacity~$B'$, where each chain~$C_i$ of~$a_i$ jobs in $\mathcal{I}'$ is represented as an item $a_i$.
        
        In order to simulate bins of different capacities within \textsc{Unary Bin Packing}, we propose an instance~$\mathcal{I}$ where $k=\ell+1$, $B = 2B'+1$ and with $k$~fill items: $r'$~fill items of size $B'+2$ and $k-r'$~fill items of size~$B'+1$. Then there is exactly one fill element per bin, and this results in $r'$ bins of remaining capacity~$B'-1$ and $k-r'$ bins of remaining capacity~$B'$. The equivalence between instances~$\mathcal{I}$ and~$\mathcal{I}'$ is then straightforward. \qed
    \end{proof}

\end{document}